\pdfoutput=1
\documentclass[11pt,reqno]{amsart}

\usepackage[margin=1.08in]{geometry}
\usepackage{amsmath,amssymb,amsthm,mathtools}
\usepackage{enumitem}
\usepackage{hyperref}
\usepackage[nameinlink,capitalise]{cleveref}

\hypersetup{colorlinks=true,linkcolor=blue,citecolor=blue,urlcolor=blue}

\makeatletter
\let\oldsqrt\sqrt
\renewcommand{\sqrt}{\@ifnextchar[\sqrt@opt\sqrt@plain}
\newcommand{\sqrt@plain}[1]{\oldsqrt{\,#1\,}}
\newcommand{\sqrt@opt}[2][]{\oldsqrt[#1]{\,#2\,}}
\newcommand{\leftappendixsections}{%
  \def\section{\@startsection{section}{1}%
    \z@{.7\linespacing\@plus\linespacing}{.5\linespacing}%
    {\normalfont\scshape\raggedright}}}
\makeatother

\theoremstyle{plain}
\newtheorem{theorem}{Theorem}[section]
\newtheorem{lemma}[theorem]{Lemma}
\newtheorem{proposition}[theorem]{Proposition}
\newtheorem{corollary}[theorem]{Corollary}
\theoremstyle{definition}
\newtheorem{definition}[theorem]{Definition}
\newtheorem{problem}[theorem]{Problem}
\newtheorem{remark}[theorem]{Remark}

\crefname{theorem}{Theorem}{Theorems}
\crefname{lemma}{Lemma}{Lemmas}
\crefname{proposition}{Proposition}{Propositions}
\crefname{corollary}{Corollary}{Corollaries}
\crefname{definition}{Definition}{Definitions}
\crefname{problem}{Problem}{Problems}
\crefname{remark}{Remark}{Remarks}

\newcommand{\C}{\mathbb C}
\newcommand{\D}{\mathsf D}
\newcommand{\Lin}{\mathsf L}
\newcommand{\Herm}{\mathsf{Herm}}
\newcommand{\CPTP}{\mathsf{CPTP}}
\newcommand{\Tr}{\operatorname{Tr}}
\newcommand{\Id}{\operatorname{Id}}

\newcommand{\CReg}{\operatorname{CReg}}
\newcommand{\Exploit}{\operatorname{Exploit}}
\newcommand{\Reg}{\operatorname{Reg}}
\newcommand{\supp}{\operatorname{supp}}
\newcommand{\eps}{\varepsilon}
\newcommand{\ot}{\otimes}
\newcommand{\ip}[2]{\left\langle #1,#2\right\rangle}
\newcommand{\norm}[1]{\left\lVert #1\right\rVert}
\newcommand{\ket}[1]{\left|#1\right\rangle}
\newcommand{\bra}[1]{\left\langle#1\right|}
\newcommand{\proj}[1]{\ket{#1}\!\bra{#1}}

\title[Coherent swap regret and channel-proof learning]{Coherent Swap Regret and\\ Channel-Proof Learning}
\author{Sohail (Neel) Sarkar}
\thanks{University of Toronto. Email: \texttt{sohail.sarkar@mail.utoronto.ca}.}
\date{May 31, 2026}

\begin{document}

\begin{abstract}
External regret only compares a learner with fixed replacement states.  In a quantum game this misses a natural physical move: a player can apply a local completely positive trace-preserving (CPTP) map to the state it actually received or prepared.  This paper studies \emph{coherent swap regret},
\[
  \CReg_T=
  \sup_{\Lambda\in\CPTP(d)}
  \sum_{t=1}^T
  \Tr\!\left[G_t\bigl(\Lambda(\rho_t)-\rho_t\bigr)\right],
\]
for played states $\rho_t\in\D(\C^d)$ and payoff effects $0\preceq G_t\preceq I$.

The main result is a clean comparison of deviation classes.  Replacement channels give the usual external-regret rate $\Theta(\sqrt{T\log d})$.  Unital channels, including unitary deviations and mixtures of unitaries, have zero minimax regret, since the maximally mixed state is fixed by all of them.  In contrast, deterministic measurement-and-preparation channels already force $\Theta(\sqrt{dT\log d})$ regret in the moderate-horizon regime $T\gtrsim d\log d$, and this rate is also sufficient for all CPTP deviations.  Thus the hard part is not coherence by itself.  It is the ability of non-unital physical rewrites to use information in the recommendation register.

The algorithm learns a deviation channel in normalized Choi form, plays a fixed point of that channel, and updates by entropic mirror ascent on the CPTP Choi slice.  Its analysis is in an oracle or finite-dimensional convex-optimization model: each round needs a relative-entropy convex program over the Choi body and a fixed point of the current channel.  The key estimate is the \emph{Variance Collapse Lemma}.  For $A_t=d(G_t\ot\rho_t^{\mathsf T})$, the mirror bound controls $\sum_t\ip{A_t^2}{J_t}$, and the trace-preserving Choi marginal $\Tr_{\rm out}J_t=I/d$ gives
\[
  \ip{A_t^2}{J_t}\le d\Tr(\rho_t^2)\le d,
\]
removing a factor $\sqrt d$ from the naive $O(d\sqrt{T\log d})$ analysis.  A purity-sensitive version replaces the worst-case variance $dT$ by $\sum_t d\Tr(\rho_t^2)$.

As an application, decentralized full-information learning in finite quantum games reaches an $\eps$-approximate separable quantum correlated equilibrium after
\[
  T=O\!\left(\max_i\frac{d_i\log d_i}{\eps^2}\right)
\]
rounds.  These equilibria are exactly channel-proof separable recommendations.  We also give an SDP audit for local CPTP exploitability, small qubit and rock-paper-scissors examples showing why replacement tests can fail, and a probing-bandit appendix with sublinear pseudo-regret under Haar-random pure-state probes.
\end{abstract}

\maketitle

\section{Introduction}

External no-regret is a useful benchmark, but it is too weak for quantum recommendations.  It only says that no fixed replacement state would have performed better.  It does not say that a player could not profit by applying a physical operation to the state it actually received.  The benchmark used here is therefore regret against \emph{local rewrites of realized behavior}, not only against fixed alternatives.

Classically, this distinction is the gap between coarse correlated equilibrium and Aumann's correlated equilibrium: after seeing its own recommendation, a player should not want to systematically transform that recommendation \cite{Aumann1974}.  The learning notions are internal regret, swap regret, and regret matching \cite{HartMasColell2000,BlumMansour2007,Ito2020}.  In a quantum game, if a player receives or produces a local quantum state $\rho_t$, a physically admissible transformation of that state is a CPTP map $\Lambda$.  Thus the quantum analogue of swap regret compares the learner's payoff to the payoff it would have received by applying the same local quantum operation to every state it actually played:
\[
  \rho_t\mapsto\Lambda(\rho_t).
\]
This paper formalizes that benchmark and gives the optimal rate in the moderate-horizon regime.

The main technical step is the \emph{Variance Collapse Lemma}.  In Choi coordinates the deviation learner receives positive gains
\[
  A_t=d(G_t\ot\rho_t^{\mathsf T}).
\]
A first-order or norm-based Choi-body analysis would pay $\norm{A_t}_\infty\le d$, leading to $O(d\sqrt{T\log d})$ regret.  The trace-preserving constraint on the learned channel gives instead
\[
  \ip{A_t^2}{J_t}
  \le d\Tr(\rho_t^2)
  \le d,
\]
so the second-order term is only $dT$.  This is where the factor $\sqrt d$ is saved, and it is why the final rate matches the classical swap-regret scale.

\subsection{The problem solved}

The first goal is to identify which physical deviation classes make internal regret hard.  Several answers look plausible at first: fixed replacement states, unitary rewrites, unital noise maps, measurement-and-preparation maps, or all CPTP maps.  The threshold turns out to be non-unital use of the recommendation register.  Unital deviations are minimax-trivial, replacement deviations have the external-regret scale, and entanglement-breaking measurement-and-preparation deviations already have the full coherent-swap rate.

Lin, Piliouras, Sim, and Varvitsiotis formulated quantum $\Phi$-equilibria and connected them to no-regret dynamics in quantum games \cite{LinPiliourasSimVarvitsiotis2024}.  In their framework, taking $\Phi$ to be all local CPTP maps yields quantum correlated equilibrium: stability against every unilateral local quantum operation.  The missing ingredient is not external regret where the learner's action is itself a channel.  It is swap regret where the learner's action is a state and the comparator is a local CPTP map applied to the learner's realized state sequence.  The following problem isolates the finite-time all-CPTP version.

\begin{problem}[CPTP internal regret for quantum states]
Given full-information payoff effects $0\preceq G_t\preceq I$ revealed after the learner chooses $\rho_t\in\D(\C^d)$, construct states such that
\[
  \sup_{\Lambda\in\CPTP(d)}
  \sum_{t=1}^T
  \Tr\!\left[G_t\bigl(\Lambda(\rho_t)-\rho_t\bigr)\right]
  =o(T),
\]
with the optimal dependence on $d$ and $T$.
\end{problem}

In the moderate-horizon regime inherited from the classical lower bound, namely $T\ge c d\log d$, the minimax rate is $\Theta(\sqrt{dT\log d})$.  For shorter fixed-accuracy regimes with very large action spaces, recent external-to-swap reductions show that the asymptotics can differ, so the optimality claim here is explicitly the moderate-horizon one.  The upper bound learns a CPTP map as a deviation rule and plays a fixed point of the current learned map.  The lower bound comes from classical swap regret after restricting payoffs and comparator channels to the diagonal classical subalgebra.  Thus even the measurement-and-preparation face of the CPTP class is already worst-case hard.

The result also shows where the difficulty comes from.  Purely unital deviation classes, including unitary channels and mixtures of unitaries, have zero minimax regret: the learner can play the maximally mixed state, which every unital deviation fixes.  Replacement-state deviations have only the ordinary external-regret scale $\Theta(\sqrt{T\log d})$.  The jump to $\Theta(\sqrt{dT\log d})$ occurs as soon as deviations may use the recommendation register through non-unital measurement-and-preparation rewrites.  So coherence alone is not the source of the lower bound.  The source is the larger physical operation class.

\subsection{Why this is a real equilibrium notion}

The equilibrium consequence is a concrete stability statement for mediated quantum recommendations.  Suppose a mediator distributes private quantum registers to players.  A rational player can insert any local physical preprocessing before the payoff device is evaluated, and finite-dimensional physical preprocessings are exactly CPTP maps.  A recommendation state that is stable only against replacement channels is therefore stable only against players who either obey the register or discard it.  It need not be stable against players who use the information in their private register.

We call this requirement \emph{channel-proofness}.  These are the same inequalities that define separable quantum correlated equilibrium, but the name emphasizes the operational meaning: a player should not gain by preprocessing its private register.  The exploitability of a proposed recommendation state is also the value of an SDP over the local Choi body.  This gives a simple audit interpretation.  External-regret dynamics may pass all replacement-deviation tests while still failing a local quantum-operation test by a constant; coherent-swap dynamics are designed to remove that failure mode.

\subsection{Main contributions}

Here is the short list of contributions.

\begin{enumerate}[leftmargin=2em]
\item We define coherent swap regret, the noncommutative internal-regret benchmark against all local CPTP maps applied to the states actually played.
\item We give an algorithm with coherent swap regret
\[
  O\!\left(\sqrt{dT\log d}\right).
\]
The analysis is a self-contained Choi-body mirror argument.  The dimension-saving step is the Variance Collapse Lemma $\ip{A_t^2}{J_t}\le d\Tr(\rho_t^2)$.  The same proof gives the purity-sensitive refinement $\CReg_T\le (2\log d)/\eta+\eta d\sum_t\Tr(\rho_t^2)$.
\item We prove a matching minimax lower bound up to constants in the standard nontrivial classical horizon regime.  The lower bound already holds for diagonal effects and deterministic entanglement-breaking deviation channels.  Conversely, if the comparator class is restricted to unital channels, coherent-swap minimax regret is exactly zero.
\item We prove decentralized convergence to approximate separable quantum correlated equilibrium at rate
\[
  O\!\left(\max_i\sqrt{\frac{d_i\log d_i}{T}}\right).
\]
\item We interpret the equilibrium inequalities as channel-proofness of quantum recommendation protocols and show that local CPTP exploitability is SDP-certifiable.  The audit applies to arbitrary finite-dimensional recommendation states, including entangled candidates, even though the decentralized learning theorem synthesizes separable ones.
\item We give a closed-form qubit audit and a diagonal rock-paper-scissors example showing that replacement stability can miss coherent deviation incentives.
\item We include a probing-bandit partial-feedback theorem.  With random pure-state probes and one Bernoulli payoff sample on probe rounds, the coherent swap pseudo-regret is $O(d^{4/3}T^{2/3}(\log d)^{1/3})$ for nonanticipating payoff effects.
\end{enumerate}

\subsection{Relation to prior work}

Quantum strategies and games have Choi-type semidefinite representations \cite{GutoskiWatrous2007}, and computational aspects of quantum equilibria have been studied by Bostanci and Watrous \cite{BostanciWatrous2022}.  Matrix-multiplicative-weights methods for quantum zero-sum games go back to Jain and Watrous \cite{JainWatrous2009}.  Recent refinements include payoff-based and bandit-style matrix multiplicative weights in quantum games \cite{LotidisMertikopoulosBambosBlanchet2023}, as well as optimistic matrix multiplicative weights for quantum zero-sum Nash computation \cite{VasconcelosEtAl2025}.  Those papers concern equilibrium computation or external no-regret dynamics over quantum states in game forms.  This paper uses a different benchmark: internal regret against local quantum operations applied to the states actually played.

Lin et al.\ \cite{LinPiliourasSimVarvitsiotis2024} develop no-regret learning and quantum $\Phi$-equilibria, including the all-CPTP deviation class that motivates this paper.  The contribution here is an optimal swap-regret rate for that all-CPTP local-deviation benchmark in the full-information, finite-dimensional state-action model.

Online learning of quantum states was studied by Aaronson, Chen, Hazan, Kale, and Nayak \cite{AaronsonChenHazanKaleNayak2018}.  Recent work studies online learning of broad families of quantum objects \cite{BansalGeorgeGhoshSikoraZheng2025}, structured quantum processes \cite{RazaCaroEisertKhatri2024}, and instance-optimal matrix multiplicative weights \cite{GongLiWangZhang2025}.  These results are adjacent but solve a different decision problem.  In particular, Bansal et al.\ give regret bounds when the learner's action is the quantum object itself, including channels.  Here the learner's action is a state, and the comparator is a channel applied to the realized sequence of states.  This is the same distinction as between external regret and swap regret in classical game theory.

Quantum correlated equilibria for shared quantum recommendations have also been studied directly, building on Zhang's quantum strategic game theory model \cite{Zhang2012}.  Deckelbaum considered players sharing entangled pure states and applying arbitrary local operations in classical complete-information games \cite{Deckelbaum2011}.  Wei and Zhang characterized quantum correlated equilibria by semidefinite conditions and gave explicit profitable local measurements when the conditions fail \cite{WeiZhang2013}; related examples show that local quantum operations can create game-theoretic advantage even when standard correlation measures do not mark the state as useful \cite{WeiZhang2017}.  The channel-proofness formulation here is in the same line of work, but the object is dynamic: it asks whether decentralized no-regret learning can synthesize states stable against all local CPTP deviations.

The fixed-point architecture is classical in spirit: Gordon, Greenwald, and Marks developed a convex-game $\Phi$-regret framework \cite{GordonGreenwaldMarks2008}, while Blum and Mansour developed the classical external-to-internal regret reduction \cite{BlumMansour2007}.  Ito proved the tight $\Omega(\sqrt{dT\log d})$ lower bound for classical swap regret in its standard nontrivial regime \cite{Ito2020}; more recent work of Dagan, Daskalakis, Fishelson, and Golowich gives improved external-to-swap reductions for large action spaces and shows that fixed-accuracy polylogarithmic-in-dimension regimes require separate interpretation \cite{DaganDaskalakisFishelsonGolowich2024}.  The matching upper and lower bounds here should be read in the same moderate-horizon regime as the classical embedded lower bound.

\section{Preliminaries}

All Hilbert spaces are finite-dimensional.  For $H\simeq\C^d$, let $\D(H)$ be the density operators on $H$ and let $\Herm(H)$ be the Hermitian operators.  The Hilbert-Schmidt pairing is $\ip{A}{B}=\Tr(AB)$ for Hermitian $A,B$.  The operator norm is $\norm{A}_\infty$ and the trace norm is $\norm{A}_1$.  Logarithms are natural.

For density operators $\rho,\sigma$ with $\supp(\rho)\subseteq\supp(\sigma)$, the quantum relative entropy is
\[
  D(\rho\|\sigma)=\Tr\rho(\log\rho-\log\sigma),
\]
and it is $+\infty$ otherwise.

\subsection{Normalized Choi representation}

Let $H_{\rm in}\simeq H_{\rm out}\simeq\C^d$ and
\[
  \ket\Omega=d^{-1/2}\sum_{i=1}^d \ket{i}_{\rm out}\ket{i}_{\rm in}.
\]
For a linear map $\Lambda:\Lin(H_{\rm in})\to\Lin(H_{\rm out})$, define its normalized Choi operator
\[
  J_\Lambda=(\Lambda\ot\Id)(\proj\Omega).
\]
Then $\Lambda$ is CPTP if and only if
\[
  J_\Lambda\succeq0,
  \qquad
  \Tr_{\rm out}J_\Lambda=I_{\rm in}/d.
\]
Consequently $\Tr J_\Lambda=1$.  Conversely, every such $J$ is the normalized Choi state of a unique CPTP map \cite{Watrous2018}.  We write
\[
  \mathcal C_d=
  \{J\in\D(H_{\rm out}\ot H_{\rm in}):\Tr_{\rm out}J=I/d\}
\]
for the CPTP Choi body.

The action of the channel is recovered as
\[
  \Lambda(\rho)=d\,\Tr_{\rm in}\bigl[(I\ot\rho^{\mathsf T})J_\Lambda\bigr].
\]
All transposes in this paper are taken in the basis defining $\ket\Omega$.  This basis dependence is standard for Choi coordinates; in particular, $\bigl(\rho^{\mathsf T}\bigr)^2=(\rho^2)^{\mathsf T}$ and $\Tr[\bigl(\rho^{\mathsf T}\bigr)^2]=\Tr(\rho^2)$.  Thus, for every Hermitian $G$,
\begin{equation}
  \Tr(G\Lambda(\rho))=
  \ip{d(G\ot\rho^{\mathsf T})}{J_\Lambda}.
  \label{eq:choi-payoff}
\end{equation}

\subsection{Coherent swap regret}

\begin{definition}[Coherent swap regret]
Given played states $\rho_t\in\D(H)$ and payoff observables $G_t\in\Herm(H)$, define
\[
  \CReg_T(\rho_{1:T};G_{1:T})=
  \sup_{\Lambda\in\CPTP(H)}
  \sum_{t=1}^T
  \Tr\!\left[G_t\bigl(\Lambda(\rho_t)-\rho_t\bigr)\right].
\]
When $0\preceq G_t\preceq I$, payoffs lie in $[0,1]$.  If $\norm{G_t}_\infty\le L$, the transformation $G_t\mapsto(G_t+LI)/(2L)$ reduces to the effect case and multiplies regret by $2L$, since CPTP maps preserve trace.
\end{definition}

The benchmark contains external regret because replacement maps $\Lambda_\sigma(\rho)=\sigma$ are CPTP\@.  It is strictly stronger because it also includes transformations depending on the state actually played.

\section{Algorithm: coherent fixed-point Choi descent}

The algorithm is stated in an oracle or finite-dimensional convex-optimization model.  The learner uses two primitives: a solver for the relative-entropy mirror step over the CPTP Choi body \(\mathcal C_d\), and a solver for a fixed point of the current channel.  Thus the theorem is a regret and equilibrium theorem for this optimization model.  It does not claim that a round has a closed-form softmax update or polylogarithmic quantum-circuit runtime.  The mirror update is the main computational step, and its KKT equations form a noncommutative matrix-scaling problem.

The learner maintains a CPTP map $\Lambda_t$ with Choi state $J_t\in\mathcal C_d$.  It then plays a fixed point of that map.

\begin{lemma}[Fixed points]
Every CPTP map $\Lambda:\Lin(H)\to\Lin(H)$ has a fixed point $\rho\in\D(H)$ satisfying $\Lambda(\rho)=\rho$.
\end{lemma}

\begin{proof}
The map $\Lambda$ is continuous and affine on the compact convex set $\D(H)$.  Brouwer's theorem gives a fixed point.  Equivalently, the Cesaro averages of the orbit of any state have limit points, and every limit point is fixed.  This is the finite-dimensional fixed-point theorem for quantum channels; see, for example, Watrous \cite[Theorem~4.25]{Watrous2018}.
\end{proof}

\begin{lemma}[Interior of mirror iterates]
\label{lem:full-rank-iterates}
Let $\mathcal K\subseteq\D(\C^m)$ be compact and convex, and suppose $P_t\in\mathcal K$ is positive definite.  If $A_t$ is bounded and
\[
  P_{t+1}=\arg\max_{P\in\mathcal K}\{\eta\ip{A_t}{P}-D(P\|P_t)\},
\]
then every maximizer $P_{t+1}$ is positive definite.  Consequently, coherent fixed-point Choi descent initialized at $I/d^2$ has full-rank exact iterates for every finite learning rate and bounded payoff sequence.
\end{lemma}

\begin{proof}
Let $P^+$ be a maximizer and suppose it has a nonzero kernel.  Since $P_t$ is positive definite and belongs to $\mathcal K$, the perturbation $Q_\alpha=(1-\alpha)P^++\alpha P_t$ is feasible for $\alpha\in(0,1)$.  Along every kernel direction of $P^+$, the perturbation opens an eigenvalue of order $\alpha$.  The entropy contribution $-\Tr(Q_\alpha\log Q_\alpha)$ therefore has one-sided derivative $+\infty$ as $\alpha\downarrow0$, whereas the linear term $\eta\ip{A_t}{Q_\alpha}$ and the term $\Tr(Q_\alpha\log P_t)$ have finite one-sided derivatives because $A_t$ and $\log P_t$ are bounded.  Hence the objective increases for sufficiently small $\alpha>0$, contradicting optimality of $P^+$.  Thus $P^+$ is positive definite.  The last claim follows by induction from $J_1=I/d^2$.
\end{proof}

\medskip
\noindent\textbf{Coherent fixed-point Choi descent.}
\begin{enumerate}[leftmargin=2em]
\item Initialize $J_1=I_{H\ot H}/d^2$, the Choi state of the completely depolarizing channel.
\item At round $t$, let $\Lambda_t$ be the channel with Choi state $J_t$.
\item Play any fixed point $\rho_t$ satisfying $\Lambda_t(\rho_t)=\rho_t$.
\item Observe the payoff effect $0\preceq G_t\preceq I$.
\item Define $A_t=d(G_t\ot\rho_t^{\mathsf T})$.
\item Update
\begin{equation}
  J_{t+1}=
  \arg\max_{J\in\mathcal C_d}
  \left\{
    \eta\ip{A_t}{J}-D(J\|J_t)
  \right\}.
  \label{eq:mirror-update}
\end{equation}
\end{enumerate}

By Lemma~\ref{lem:full-rank-iterates}, the exact iterates remain full-rank for finite $\eta$ and bounded $A_t$.

The fixed-point condition gives
\begin{equation}
  \ip{A_t}{J_t}
  =\Tr(G_t\Lambda_t(\rho_t))
  =\Tr(G_t\rho_t).
  \label{eq:fixed-point-payoff}
\end{equation}
For any comparator channel $\Lambda^\star$ with Choi state $J^\star$,
\begin{equation}
  \ip{A_t}{J^\star}
  =\Tr(G_t\Lambda^\star(\rho_t)).
  \label{eq:comparator-payoff}
\end{equation}
Therefore coherent swap regret against $\Lambda^\star$ is exactly external regret of the Choi learner:
\[
  \sum_{t=1}^T\Tr\!\left[G_t\bigl(\Lambda^\star(\rho_t)-\rho_t\bigr)\right]
  =\sum_{t=1}^T\ip{A_t}{J^\star-J_t}.
\]

\begin{remark}[Approximate fixed points]
If $\norm{\Lambda_t(\rho_t)-\rho_t}_1\le\delta_t$, then
\[
  \left|\Tr(G_t\Lambda_t(\rho_t))-\Tr(G_t\rho_t)\right|\le\delta_t
\]
for $0\preceq G_t\preceq I$.  The coherent-swap-regret bound acquires the additive error $\sum_t\delta_t$.  Therefore average equilibrium guarantees require vanishing average fixed-point error: in the $T$-round equilibrium theorem the error term is $T^{-1}\sum_t\delta_t$.  A sufficient implementation condition for an $\eps$-equilibrium is to choose $T$ so that the regret term is at most $\eps/2$ and compute each fixed point to accuracy $\delta_t\le\eps/2$; equivalently, up to constants, one needs both $T=O(d\log d/\eps^2)$ and per-round fixed-point error $O(\eps)$.
\end{remark}

\section{A self-contained Choi-body mirror analysis}
\label{sec:second-order}

This section proves the second-order matrix-multiplicative-weights estimate in the form used later.  The proof is included for completeness.  It is a direct one-step mirror argument for arbitrary convex quantum decision sets, using the Gibbs variational principle and the Golden-Thompson inequality.  The genuinely problem-specific step comes later, when the CPTP Choi constraint gives the variance collapse of Lemma~\ref{lem:variance-collapse}.

\begin{theorem}[Second-order mirror inequality for positive gains]
\label{thm:second-order-md}
Let $\mathcal K\subseteq\D(\C^m)$ be compact and convex.  Let $P_1\in\mathcal K$ be positive definite.  Let $A_1,\ldots,A_T\succeq0$ satisfy $\eta\norm{A_t}_\infty\le1$ for every $t$, and define
\[
  P_{t+1}=
  \arg\max_{P\in\mathcal K}
  \left\{
    \eta\ip{A_t}{P}-D(P\|P_t)
  \right\}.
\]
Then, for every $U\in\mathcal K$,
\[
  \sum_{t=1}^T\ip{A_t}{U-P_t}
  \le
  \frac{D(U\|P_1)}{\eta}
  +\eta\sum_{t=1}^T\ip{A_t^2}{P_t}.
\]
\end{theorem}

\begin{proof}
Fix a round and write $P=P_t$, $P^+=P_{t+1}$, and $A=A_t$.  By Lemma~\ref{lem:full-rank-iterates}, $P$ and $P^+$ are positive definite, so the logarithms and first-order conditions below are well-defined.

Let
\[
  \Phi(Q)=\eta\ip{A}{Q}-D(Q\|P).
\]
The Fr\'echet derivative at $P^+$ is
\[
  \nabla\Phi(P^+)=\eta A-\log P^+-I+\log P.
\]
First-order optimality for maximizing the differentiable concave function $\Phi$ over the convex set $\mathcal K$ gives the variational inequality
\[
  \ip{\eta A-\log P^+-I+\log P}{U-P^+}\le0
  \qquad\forall U\in\mathcal K.
\]
Equivalently, because $\Tr(U-P^+)=0$,
\begin{align*}
  \eta\ip{A}{U-P^+}
  &\le \ip{\log P^+-\log P}{U-P^+} \\
  &=D(U\|P)-D(U\|P^+)-D(P^+\|P).
\end{align*}
For the Choi slice $\mathcal C_d$, this is exactly the KKT multiplier calculation: positivity is inactive at the full-rank optimum, and the equality constraints give
\[
  \eta A-\log P^+-I+\log P
  =\gamma I+I_{\rm out}\ot B
\]
for a scalar $\gamma$ and a Hermitian multiplier $B$.  If $U,P^+\in\mathcal C_d$, then
\[
  \ip{\gamma I}{U-P^+}=\gamma(\Tr U-\Tr P^+)=0
\]
and
\[
  \ip{I_{\rm out}\ot B}{U-P^+}
  =\Tr\!\left[B\bigl(\Tr_{\rm out}U-\Tr_{\rm out}P^+\bigr)\right]=0,
\]
so the multipliers cancel on all feasible differences.

Thus
\begin{align}
  \eta\ip{A}{U-P}
  &\le
  D(U\|P)-D(U\|P^+)
  +\eta\ip{A}{P^+-P}-D(P^+\|P).
  \label{eq:threepoint-gap}
\end{align}
It remains to control the last two terms.

The Gibbs variational principle says that, for Hermitian $H$,
\[
  \log\Tr e^H
  =\sup_{Q\in\D(\C^m)}\{\ip{H}{Q}-\Tr(Q\log Q)\}.
\]
With $H=\log P+\eta A$, for every density matrix $Q$,
\[
  \eta\ip{A}{Q}-D(Q\|P)
  \le
  \log\Tr\exp(\log P+
  \eta A).
\]
Taking $Q=P^+$ and subtracting $\eta\ip{A}{P}$ gives
\begin{equation}
  \eta\ip{A}{P^+-P}-D(P^+\|P)
  \le
  \log\Tr\exp(\log P+
  \eta A)-\eta\ip{A}{P}.
  \label{eq:mgf-reduction}
\end{equation}
By Golden-Thompson \cite{Carlen2010},
\[
  \Tr\exp(\log P+
  \eta A)\le\Tr(Pe^{\eta A}).
\]
Since $0\preceq\eta A\preceq I$, the scalar inequality $e^x\le1+x+x^2$ on $[0,1]$ gives the operator inequality
\[
  e^{\eta A}\preceq I+\eta A+
  \eta^2A^2.
\]
Hence
\[
  \Tr\exp(\log P+
  \eta A)
  \le 1+
  \eta\ip{A}{P}+
  \eta^2\ip{A^2}{P}.
\]
Using $\log(1+x+y)\le x+y$ for $x,y\ge0$,
\[
  \log\Tr\exp(\log P+
  \eta A)-\eta\ip{A}{P}
  \le
  \eta^2\ip{A^2}{P}.
\]
Combining this with \eqref{eq:threepoint-gap} and dividing by $\eta$ yields
\[
  \ip{A_t}{U-P_t}
  \le
  \frac{D(U\|P_t)-D(U\|P_{t+1})}{\eta}
  +
  \eta\ip{A_t^2}{P_t}.
\]
Summing over $t$ telescopes the relative entropies and proves the theorem.
\end{proof}

\begin{remark}
The proof uses only the mirror update displayed in the theorem.  If one rewrites the same exact iterates in a regularized-leader form, the regularizer must be $R(P)=D(P\|P_1)$ and the initialization must be the minimizer $P_1$; this equivalence is not used in the analysis.  The theorem itself is a mirror-descent statement.
\end{remark}

\section{Optimal coherent swap regret}
\label{sec:main-upper}

\begin{lemma}[Choi entropy radius]
\label{lem:choi-radius}
For every $J\in\mathcal C_d$,
\[
  D\left(J\middle\|\frac{I}{d^2}\right)\le2\log d.
\]
The bound is tight for unitary channels.
\end{lemma}

\begin{proof}
Since $I/d^2$ is maximally mixed on $H\ot H$,
\[
  D\left(J\middle\|\frac{I}{d^2}\right)=-S(J)+2\log d\le2\log d.
\]
For a unitary channel, $J$ is a pure maximally entangled Choi state, so $S(J)=0$.
\end{proof}

\begin{lemma}[Variance collapse]
\label{lem:variance-collapse}
Let $J\in\mathcal C_d$, $\rho\in\D(H)$, and $0\preceq G\preceq I$.  Set
\[
  A=d(G\ot\rho^{\mathsf T}).
\]
Then $A\succeq0$, $\norm{A}_\infty\le d$, and
\[
  \ip{A^2}{J}\le d\Tr(\rho^2)\le d.
\]
\end{lemma}

\begin{proof}
Clearly $A\succeq0$ and $\norm{A}_\infty\le d$.  Since $G^2\preceq I$,
\[
  A^2=d^2\bigl(G^2\ot\bigl(\rho^{\mathsf T}\bigr)^2\bigr)
  \preceq d^2\bigl(I\ot\bigl(\rho^{\mathsf T}\bigr)^2\bigr).
\]
Using $\Tr_{\rm out}J=I/d$,
\begin{align*}
  \ip{A^2}{J}
  &\le d^2\Tr\!\left[J\bigl(I\ot\bigl(\rho^{\mathsf T}\bigr)^2\bigr)\right] \\
  &=d^2\Tr\!\left[(\Tr_{\rm out}J)\bigl(\rho^{\mathsf T}\bigr)^2\right] \\
  &=d\Tr(\rho^2)\le d.
\end{align*}
\end{proof}

\paragraph{Where the dimension is saved.}
If one used only the norm bound $\norm{A_t}_\infty\le d$, the second-order term in \cref{thm:second-order-md} would be bounded as
\[
  \ip{A_t^2}{J_t}\le \norm{A_t}_\infty^2\Tr J_t\le d^2,
\]
which would give
\[
  \CReg_T\le \frac{2\log d}{\eta}+\eta d^2T
  =O\!\left(d\sqrt{T\log d}\right).
\]
The variance collapse replaces $d^2$ by $d\Tr(\rho_t^2)\le d$.  The reason is that the Choi iterate is not an arbitrary density matrix on $H\otimes H$; it satisfies the trace-preserving marginal constraint $\Tr_{\rm out}J_t=I/d$.  This is where the CPTP constraint improves the dimension dependence, and it is why the final rate matches the classical swap-regret scale.

\begin{theorem}[Coherent swap regret upper bound]
\label{thm:upper}
Let $d\ge2$ and let $0\preceq G_t\preceq I$ be arbitrary.  Coherent fixed-point Choi descent with any $\eta\le1/d$ satisfies
\[
  \CReg_T\le\frac{2\log d}{\eta}+\eta dT.
\]
If $T\ge2d\log d$ and $\eta=\sqrt{2\log d/(dT)}$, then
\[
  \CReg_T\le2\sqrt{2dT\log d}.
\]
For observables $H_t$ with $\norm{H_t}_\infty\le1$, applying the algorithm to $(H_t+I)/2$ gives regret at most $4\sqrt{2dT\log d}$.
\end{theorem}

\begin{proof}
Fix a comparator channel $\Lambda^\star$ with Choi state $J^\star$.  By \eqref{eq:fixed-point-payoff} and \eqref{eq:comparator-payoff},
\[
  \sum_{t=1}^T\Tr\!\left[G_t(\Lambda^\star(\rho_t)-\rho_t)\right]
  =\sum_{t=1}^T\ip{A_t}{J^\star-J_t}.
\]
Since $\eta\le1/d$ and $\norm{A_t}_\infty\le d$, \cref{thm:second-order-md} applies.  Thus
\[
  \sum_{t=1}^T\ip{A_t}{J^\star-J_t}
  \le
  \frac{D(J^\star\|J_1)}{\eta}
  +\eta\sum_{t=1}^T\ip{A_t^2}{J_t}.
\]
Here $J_1=I/d^2$.  By Lemma~\ref{lem:choi-radius}, the entropy term is at most $2\log d$, and by Lemma~\ref{lem:variance-collapse}, each variance term is at most $d$.  Taking the supremum over $\Lambda^\star$ proves
\[
  \CReg_T\le\frac{2\log d}{\eta}+\eta dT.
\]
The displayed choice of $\eta$ is feasible when $T\ge2d\log d$ and gives $2\sqrt{2dT\log d}$.  The norm-bounded observable claim follows by the shift-and-rescale argument in the definition of coherent swap regret.
\end{proof}

\begin{theorem}[Purity-sensitive coherent swap regret]
\label{thm:purity-sensitive}
For a run of coherent fixed-point Choi descent, define
\[
  V_T=\sum_{t=1}^T d\Tr(\rho_t^2),
  \qquad
  B_T=\max_{t\le T}d\norm{\rho_t}_\infty.
\]
For every comparator channel and every learning rate $\eta\le1/B_T$,
\[
  \CReg_T\le \frac{2\log d}{\eta}+\eta V_T.
\]
Consequently, if $V_T$ and $B_T$ are known in advance and
\[
  \eta=\min\left\{\frac1{B_T},\sqrt{\frac{2\log d}{V_T}}\right\},
\]
then
\[
  \CReg_T\le 2\sqrt{2V_T\log d}+4B_T\log d.
\]
Maximally mixed fixed-point play has $V_T=T$ and $B_T=1$, giving the external-regret scale $O(\sqrt{T\log d}+\log d)$; pure fixed-point play has $V_T=dT$ and recovers the worst-case $O(\sqrt{dT\log d})$ rate.
\end{theorem}

\begin{proof}
For a comparator Choi state $J^\star$,
\[
  \sum_{t=1}^T\Tr\!\left[G_t(\Lambda^\star(\rho_t)-\rho_t)\right]
  =\sum_{t=1}^T\ip{A_t}{J^\star-J_t}.
\]
The second-order mirror inequality gives
\[
  \sum_{t=1}^T\ip{A_t}{J^\star-J_t}
  \le
  \frac{D(J^\star\|J_1)}{\eta}
  +\eta\sum_{t=1}^T\ip{A_t^2}{J_t}.
\]
The entropy radius is at most $2\log d$, and Lemma~\ref{lem:variance-collapse} gives $\ip{A_t^2}{J_t}\le d\Tr(\rho_t^2)$, yielding the fixed-rate bound.  The step-size condition in Theorem~\ref{thm:second-order-md} follows from
\[
  \norm{A_t}_\infty\le d\norm{\rho_t}_\infty\le B_T.
\]
Optimizing with the displayed clipped learning rate gives the stated bound.  If the unconstrained optimum is feasible, the value is $2\sqrt{2V_T\log d}$.  Otherwise $V_T<2B_T^2\log d$, and the choice $\eta=1/B_T$ gives a bound at most $4B_T\log d$, which is absorbed by the displayed expression.
\end{proof}

\begin{proposition}[Unital comparator classes have zero minimax regret]
\label{prop:unital-zero}
Let $\mathsf U_d$ be any class of unital CPTP maps on $\C^d$, meaning $\Lambda(I)=I$ for every $\Lambda\in\mathsf U_d$.  If coherent swap regret is defined with the supremum restricted to $\mathsf U_d$, then the minimax regret is exactly zero for every horizon and every payoff sequence.
\end{proposition}

\begin{proof}
The learner plays $\rho_t=I/d$ at every round.  For every unital channel $\Lambda$,
\[
  \Lambda(\rho_t)=\Lambda(I/d)=I/d=\rho_t.
\]
Therefore every summand $\Tr[G_t(\Lambda(\rho_t)-\rho_t)]$ is zero.  Nonnegativity of minimax regret gives equality.
\end{proof}

\begin{proposition}[Replacement deviations have the external-regret scale]
\label{prop:replacement-scale}
Let $\mathfrak R_T^{\rm repl}(d)$ denote the minimax regret when the comparator class is restricted to replacement channels $\Lambda_\sigma(\rho)=\sigma$.  There is a universal constant $c>0$ such that, for $T\ge c\log d$,
\[
  \mathfrak R_T^{\rm repl}(d)=\Theta(\sqrt{T\log d}).
\]
\end{proposition}

\begin{proof}
The restricted benchmark is exactly external regret over density matrices:
\[
  \sup_{\sigma\in\D(\C^d)}\sum_{t=1}^T\Tr[G_t(\sigma-\rho_t)].
\]
The upper bound is the standard entropic matrix-weights bound on the $d$-dimensional spectraplex: the entropy radius around $I/d$ is at most $\log d$, and for effects $0\preceq G_t\preceq I$ the second-order term is at most $T$.  The lower bound restricts to diagonal states and diagonal effects, reducing to prediction with $d$ experts.
\end{proof}

The unital proposition is a useful sanity check on the lower bound.  The hard deviations are not arbitrary physical noise maps; they must include non-unital transformations that can systematically rewrite a recommendation into a more rewarding state.  The diagonal classical modification maps used in the lower bound are exactly of this non-unital kind.

\begin{theorem}[Minimax lower bound in the moderate-horizon regime]
\label{thm:lower}
Let
\[
  \mathfrak R_T^{\rm coh}(d)
  :=
  \inf_{\mathcal A}\sup_{G_{1:T}}
  \mathbb E\,\CReg_T(\mathcal A;G_{1:T})
\]
be the expected minimax coherent-swap regret for payoff effects $0\preceq G_t\preceq I$ on $\D(\C^d)$.  There are universal constants $c_0,c_1>0$ such that, for all $d\ge2$ and all $T\ge c_0d\log d$,
\[
  \mathfrak R_T^{\rm coh}(d)
  \ge
  c_1\sqrt{dT\log d}.
\]
Since \cref{thm:upper} is deterministic for every adversarial sequence, it also upper-bounds the expected regret of any randomized implementation.  Thus, for $T\ge C d\log d$ with a sufficiently large universal constant $C$,
\[
  \mathfrak R_T^{\rm coh}(d)
  =\Theta(\sqrt{dT\log d}).
\]
The restriction $T\gtrsim d\log d$ is part of the statement, not a hidden assumption: it is the regime in which the embedded classical lower bound is the relevant minimax obstruction and in which the displayed upper-bound tuning satisfies $\eta\le1/d$.
\end{theorem}

\begin{proof}
Restrict the adversary to diagonal effects
\[
  G_t=\sum_{i=1}^d g_{t,i}\proj{i},\qquad g_{t,i}\in[0,1].
\]
A quantum state $\rho_t$ then induces the classical mixed action $p_t(i)=\langle i|\rho_t|i\rangle$, and $\Tr(G_t\rho_t)=\sum_i p_t(i)g_{t,i}$.  Every classical modification rule $\phi:[d]\to[d]$ is represented by the CPTP map
\[
  \Lambda_\phi(X)=\sum_{i=1}^d\langle i|X|i\rangle\proj{\phi(i)}.
\]
For this map,
\[
  \Tr(G_t\Lambda_\phi(\rho_t))=\sum_i p_t(i)g_{t,\phi(i)}.
\]
Thus the coherent-swap problem contains the classical swap-regret problem as a diagonal subgame.  Ito's classical lower bound gives $\Omega(\sqrt{dT\log d})$ in the moderate-horizon range stated above \cite{Ito2020}.  The same lower bound therefore applies to the quantum problem.  This lower bound already comes from an entanglement-breaking, diagonal subproblem.  It proves worst-case optimality in the stated regime without using any genuinely noncommutative adversarial construction.
\end{proof}

\begin{corollary}[Hardness already for entanglement-breaking deviations]
\label{cor:eb-lower}
The lower bound in \cref{thm:lower} remains valid if the adversary is restricted to diagonal payoff effects and the comparator class in the regret benchmark is restricted to deterministic measurement-and-preparation channels
\[
  \Lambda_\phi(X)=\sum_{i=1}^d\langle i|X|i\rangle\proj{\phi(i)},
  \qquad \phi:[d]\to[d].
\]
In particular, the worst-case obstruction is already present inside a strict entanglement-breaking face of the CPTP deviation class.
\end{corollary}

\begin{proof}
The proof of \cref{thm:lower} uses only diagonal effects and only the channels displayed above.  These channels are entanglement-breaking because they first measure in the computational basis and then prepare \(\proj{\phi(i)}\) conditioned on the outcome.  Therefore the same classical swap-regret lower bound applies under the stated restrictions.
\end{proof}

Putting Proposition~\ref{prop:replacement-scale}, Proposition~\ref{prop:unital-zero}, and Corollary~\ref{cor:eb-lower} together gives a three-level picture.  Replacement deviations recover ordinary external regret, unital deviations are minimax-trivial, and entanglement-breaking measurement-and-preparation deviations already have the full coherent-swap rate.  This answers a natural question: the lower bound is not driven by coherent unitary effects.  In the minimax model, the difficulty is the physically allowed non-unital use of the recommendation register.

The lower bound is intentionally embedded in a classical subproblem, because that already proves worst-case optimality in the stated regime.  It remains open whether genuinely quantum structure, such as non-entanglement-breaking comparator channels or restricted noncommutative adversaries, gives sharper structure-dependent lower bounds.

The regime qualifier is also important.  Recent external-to-swap reductions show different behavior for fixed-accuracy, polylogarithmic-in-dimension round regimes \cite{DaganDaskalakisFishelsonGolowich2024}.  The matching statement here is only for $T\ge c d\log d$.  We do not claim that $\sqrt{dT\log d}$ is the right asymptotic law in every small-horizon or fixed-accuracy large-action regime.

\section{Learning separable quantum correlated equilibria}

There are $n$ players.  Player $i$ controls $H_i\simeq\C^{d_i}$.  A product play profile is
\[
  \rho_t=\rho_{1,t}\ot\cdots\ot\rho_{n,t}.
\]
Player $i$ has payoff effect $0\preceq U_i\preceq I$ on $H_1\ot\cdots\ot H_n$, and payoff $u_i(\rho)=\Tr(U_i\rho)$.

\begin{definition}[Channel-proof quantum correlated equilibrium]
A state $\omega$ on $H_1\ot\cdots\ot H_n$, possibly entangled, is an $\eps$-approximate channel-proof quantum correlated equilibrium if, for every player $i$ and every local CPTP map $\Lambda_i$,
\[
  \Tr(U_i\omega)
  \ge
  \Tr\!\left[U_i(\Lambda_i\ot\Id_{-i})(\omega)\right]-\eps.
\]
If $\omega$ is separable, we call it an $\eps$-approximate separable quantum correlated equilibrium.
\end{definition}

\begin{theorem}[Finite-time decentralized computation]
\label{thm:equilibrium}
Suppose every player runs coherent fixed-point Choi descent using its local payoff feedback.  Let
\[
  \bar\rho_T=\frac1T\sum_{t=1}^T\rho_{1,t}\ot\cdots\ot\rho_{n,t}.
\]
If $T\ge2d_i\log d_i$ for every $i$ and the fixed points are exact, then $\bar\rho_T$ is an $\eps_T$-approximate separable quantum correlated equilibrium with
\[
  \eps_T\le\max_i 2\sqrt{\frac{2d_i\log d_i}{T}}.
\]
If player $i$ instead uses approximate fixed points with errors $\delta_{i,t}$ in trace norm, the right-hand side becomes
\[
  \eps_T\le\max_i\left(2\sqrt{\frac{2d_i\log d_i}{T}}+\frac1T\sum_{t=1}^T\delta_{i,t}\right).
\]
Thus $T=O(\max_i d_i\log d_i/\eps^2)$ rounds plus average fixed-point error at most $O(\eps)$ suffices for an $\eps$-approximate equilibrium in the full-information model; a simple sufficient condition is $\delta_{i,t}\le O(\eps)$ for every $i,t$.
\end{theorem}

\begin{proof}
Fix player $i$.  At round $t$, define $G_{i,t}$ by
\[
  \Tr(G_{i,t}\sigma_i)=
  \Tr\!\left[U_i(\sigma_i\ot\rho_{-i,t})\right]
  \qquad\forall\sigma_i\in\D(H_i).
\]
Since $0\preceq U_i\preceq I$, we have $0\preceq G_{i,t}\preceq I$.  Applying \cref{thm:upper} to player $i$, for every local CPTP map $\Lambda_i$,
\[
  \sum_{t=1}^T\Tr(G_{i,t}\rho_{i,t})
  \ge
  \sum_{t=1}^T\Tr(G_{i,t}\Lambda_i(\rho_{i,t}))
  -2\sqrt{2d_i T\log d_i}.
\]
Equivalently,
\[
  \sum_{t=1}^T u_i(\rho_t)
  \ge
  \sum_{t=1}^T\Tr\!\left[U_i(\Lambda_i\ot\Id_{-i})(\rho_t)\right]
  -2\sqrt{2d_i T\log d_i}.
\]
Divide by $T$ and use linearity to obtain
\[
  \Tr(U_i\bar\rho_T)
  \ge
  \Tr\!\left[U_i(\Lambda_i\ot\Id_{-i})(\bar\rho_T)\right]
  -2\sqrt{\frac{2d_i\log d_i}{T}}.
\]
This holds for all $i$ and all $\Lambda_i$.  The state $\bar\rho_T$ is separable because it is an average of product states.  If player $i$ uses approximate fixed points, the approximate-fixed-point remark adds $\sum_t\delta_{i,t}$ to player $i$'s cumulative coherent swap regret.  Dividing by $T$ gives the displayed additional average error.
\end{proof}

\begin{remark}[Tightness of the equilibrium-rate reduction]
The rate in Theorem~\ref{thm:equilibrium} is tight for this no-regret route to equilibrium.  The diagonal subgames used in Theorem~\ref{thm:lower} are valid quantum games, and an opponent's current state can induce the diagonal local payoff effects used in the embedded classical swap-regret problem.  Thus any analysis that reaches separable quantum correlated equilibrium only by bounding each player's coherent swap regret cannot improve the dependence $\sqrt{d_i\log d_i/T}$ in the moderate-horizon regime.  This is not a lower bound on arbitrary offline computation of an equilibrium from an explicit payoff tensor; that would be a different question.
\end{remark}

\section{Application: channel-proof quantum recommendation protocols}
\label{sec:application}

This section spells out what \cref{thm:equilibrium} computes.  In a mediated quantum recommendation protocol, each player receives a private register.  The protocol is credible only if a player cannot gain by applying a local quantum operation to that register before the payoff interaction.

\begin{definition}[Channel-proof recommendation]
A state $\omega$ is an $\eps$-channel-proof recommendation protocol for payoff effects $U_i$ if, for every player $i$ and every local CPTP map $\Lambda_i$,
\[
  \Tr(U_i\omega)
  \ge
  \Tr\!\left[U_i(\Lambda_i\ot\Id_{-i})(\omega)\right]-\eps.
\]
\end{definition}

For separable $\omega$, this definition matches the separable quantum correlated-equilibrium inequalities in the previous section.  Channel-proofness is just the operational reading of those inequalities when the player can preprocess its private quantum register.  The audit definition also applies to arbitrary finite-dimensional recommendation states, including entangled candidates.

\begin{proposition}[SDP audit of CPTP exploitability]
\label{prop:sdp-audit}
Fix a candidate finite-dimensional recommendation state $\omega$, not necessarily separable.  For each player $i$, there is an explicitly computable Hermitian operator $B_i(\omega)$ on $H_i^{\rm out}\ot H_i^{\rm in}$ such that
\[
  \Tr\!\left[U_i(\Lambda_i\ot\Id_{-i})(\omega)\right]
  =\ip{B_i(\omega)}{J_{\Lambda_i}}
\]
for every local channel $\Lambda_i$.  Hence the local CPTP exploitability of $\omega$ is exactly
\[
  \Exploit_i(\omega)=
  \max_{J\succeq0,\,\Tr_{\rm out}J=I/d_i}
  \ip{B_i(\omega)}{J}-\Tr(U_i\omega),
\]
a semidefinite program over the local Choi body.
\end{proposition}

\begin{proof}
Fix $i$ and $\omega$.  By the normalized Choi representation,
\[
  \Lambda_i(X)=d_i\Tr_{\rm in}\!\left[(I_{\rm out}\ot X^{\mathsf T})J_{\Lambda_i}\right].
\]
Consequently the map
\[
  J_{\Lambda_i}\mapsto
  \Tr\!\left[U_i(\Lambda_i\ot\Id_{-i})(\omega)\right]
\]
is linear in $J_{\Lambda_i}$.  Finite-dimensional Hilbert-Schmidt duality gives the Hermitian representative $B_i(\omega)$; it can be computed explicitly by expanding $\omega$ and $U_i$ in any matrix-unit basis and collecting the coefficients multiplying $J_{\Lambda_i}$.  Maximizing over all CPTP maps is therefore the linear optimization problem over the spectrahedron
\[
  J\succeq0,
  \qquad
  \Tr_{\rm out}J=I/d_i.
\]
Subtracting the baseline payoff gives the displayed exploitability value.
\end{proof}

The SDP is ``tractable'' only in the finite-dimensional convex-optimization sense.  The variable is a $d_i^2\times d_i^2$ Choi matrix, so the audit is polynomial-time in the explicit matrix dimension and input bit complexity, but it may still be expensive for large local Hilbert spaces.  Its dual is often the most useful certificate:
\begin{equation}
  \tag{D}
  \begin{aligned}
  \text{minimize}\quad& \frac1{d_i}\Tr Y\\
  \text{subject to}\quad& I_{\rm out}\ot Y\succeq B_i(\omega),\\
  &Y=Y^*.
  \end{aligned}
  \label{eq:sdp-dual}
\end{equation}
Strong duality holds because the local Choi body has an interior feasible point, for example $J=I/d_i^2$.

\subsection{A closed-form qubit audit}

The SDP audit is not only a formal certificate.  Even in dimension two it detects coherent preprocessing gains that replacement-state tests underestimate.

Let $H_A=H_B=\C^2$, let $\ket{+}=(\ket0+\ket1)/\sqrt2$, and consider the two-register recommendation
\[
  \omega=\frac12\proj{0}_A\ot\proj{0}_B+\frac12\proj{+}_A\ot\proj{1}_B.
\]
Player $A$'s payoff effect is
\[
  U=\proj{+}_A\ot\proj{0}_B+\proj{1}_A\ot\proj{1}_B.
\]
Thus the second register is a classical context: when $B=0$, $A$ is rewarded for becoming $\ket+$; when $B=1$, $A$ is rewarded for becoming $\ket1$.

\begin{proposition}[Qubit audit: coherent preprocessing beats replacement]
\label{prop:qubit-audit}
For the recommendation state $\omega$ and payoff effect $U$ above,
\[
  \Tr(U\omega)=\frac12.
\]
The optimal local CPTP deviation for player $A$ attains payoff $1$, so the CPTP exploitability is exactly $1/2$.  By contrast, the best replacement channel $\rho\mapsto\sigma$ attains payoff
\[
  \frac12\left(1+\frac1{\sqrt2}\right),
\]
so replacement-channel exploitability is only $1/(2\sqrt2)$.
\end{proposition}

\begin{proof}
The baseline is
\[
  \Tr(U\omega)=\frac12\left(|\langle +|0\rangle|^2+|\langle 1|+\rangle|^2\right)=\frac12.
\]
Let
\[
  V=\frac1{\sqrt2}\begin{pmatrix}1&-1\\[2pt]1&1\end{pmatrix}.
\]
Then $V\ket0=\ket+$ and $V\ket+=\ket1$.  The unitary channel $\Lambda_V(X)=VXV^\dagger$ therefore maps both recommended states to their context-dependent rewarded states, giving
\[
  \Tr\!\bigl[U(\Lambda_V\ot\Id)(\omega)\bigr]=1.
\]
Since $0\preceq U\preceq I$, no channel can obtain payoff larger than $1$, so the SDP optimum is exactly $1$ and the CPTP exploitability is $1/2$.

For a replacement channel $\Lambda_\sigma(\rho)=\sigma$, the payoff is
\[
  \frac12\Tr\!\left[(\proj{+}+\proj{1})\sigma\right].
\]
The largest eigenvalue of $\proj{+}+\proj{1}$ is $1+|\langle +|1\rangle|=1+1/\sqrt2$.  Optimizing over density matrices $\sigma$ gives the displayed replacement payoff and exploitability.  Hence the coherent audit detects a strictly larger deviation incentive than any replacement-state audit.
\end{proof}

This example shows the issue in the smallest possible setting.  The problem is not entanglement or large dimension.  A recommendation can contain local state information that a rational player may process before acting.  Replacement tests ask only whether the player should discard the register.  Channel-proofness asks whether the player can profit by using it.

The same audit definition applies to entangled recommendation states.  If a mediator distributes an arbitrary state $\omega$ rather than a separable time average, channel-proofness still means nonpositivity of every local CPTP exploitability SDP\@.  What changes is the synthesis theorem: uncoupled product play can only produce separable averages, so \cref{thm:equilibrium} is a decentralized learning theorem for the separable part of the recommendation set, while Proposition~\ref{prop:sdp-audit} is an audit theorem for any fixed finite-dimensional recommendation state.

This is the role of coherent swap regret in the recommendation problem.  It gives dynamics that synthesize channel-proof recommendations, and it also gives a way to audit a proposed recommendation state.  External-regret dynamics certify only replacement-channel stability.  The next section shows that replacement stability can fail the channel-proof test by a constant.

\section{Why external regret is insufficient}

\subsection{Coarse stability can fail channel-proofness}

Quantum coarse correlated equilibrium tests replacement channels.  Quantum correlated equilibrium tests all local CPTP maps.  The following diagonal example shows that the difference is operationally large.

Consider the zero-sum rock-paper-scissors payoff matrix for player 1.  This subsection uses the conventional $[-1,1]$ payoff normalization; shifting and rescaling to effects in $[0,I]$ only rescales the displayed exploitability gap.
\[
  M=
  \begin{pmatrix}
  0&-1&1\\
  1&0&-1\\
  -1&1&0
  \end{pmatrix},
\]
with player 2 payoff $-M$.  Embed the game diagonally on $\C^3\ot\C^3$:
\[
  U_1=\sum_{a,b\in\{R,P,S\}}M_{ab}\proj{a}\ot\proj{b},
  \qquad U_2=-U_1.
\]
Let
\[
  \omega=\frac13\sum_{a\in\{R,P,S\}}\proj{a}\ot\proj{a}.
\]

\begin{proposition}[Coarse stability does not imply coherent stability]
\label{prop:rps-separation}
The state $\omega$ is a separable quantum coarse correlated equilibrium, but player 1 has a local CPTP deviation that improves its payoff by exactly $1$.
\end{proposition}

\begin{proof}
The payoff at $\omega$ is zero because every realized pair is a tie.  If either player uses a replacement channel, its new local state is independent of the other player's register.  The other player's marginal under $\omega$ is uniform over rock, paper, scissors.  Against the uniform action, every fixed mixed action has expected payoff zero.  Thus no replacement deviation improves payoff, so $\omega$ is a separable quantum coarse correlated equilibrium.

Let $\phi$ be the winning-response map
\[
  \phi(R)=P,
  \qquad
  \phi(P)=S,
  \qquad
  \phi(S)=R.
\]
The channel
\[
  \Lambda_\phi(X)=\sum_{a\in\{R,P,S\}}\langle a|X|a\rangle\proj{\phi(a)}
\]
is CPTP\@.  Applying it to player 1's register transforms every tied pair $(a,a)$ into a winning pair $(\phi(a),a)$, so
\[
  \Tr\!\left[U_1(\Lambda_\phi\ot\Id)(\omega)\right]=1,
  \qquad
  \Tr(U_1\omega)=0.
\]
Hence the local CPTP exploitability is exactly $1$ for player 1.
\end{proof}

\subsection{Zero external regret and linear coherent swap regret}

The static example separates equilibrium notions.  The next example separates the online regret notions directly.

\begin{theorem}[Qubit cyclic separation]
\label{thm:qubit-separation}
For every $T$ divisible by $4$, there exist qubit states $\rho_t$ and observables $H_t$ with $\norm{H_t}_\infty\le1$ such that
\[
  \sup_{\sigma\in\D(\C^2)}
  \sum_{t=1}^T\Tr[H_t(\sigma-\rho_t)]=0,
\]
but
\[
  \sup_{\Lambda\in\CPTP(2)}
  \sum_{t=1}^T\Tr[H_t(\Lambda(\rho_t)-\rho_t)]=T.
\]
\end{theorem}

\begin{proof}
Let $\sigma_x,\sigma_y,\sigma_z$ be the Pauli matrices and set
\[
  v_k=(\cos(k\pi/2),\sin(k\pi/2),0),\qquad k=0,1,2,3.
\]
Let $\rho_k=(I+v_k\cdot\sigma)/2$ and cycle through $\rho_t=\rho_{t\bmod4}$.  Define
\[
  H_t=v_{t+1\bmod4}\cdot\sigma.
\]
Then $\Tr(H_t\rho_t)=v_{t+1}\cdot v_t=0$ each round.  For any fixed comparator state $\sigma_r=(I+r\cdot\sigma)/2$ with $\norm r_2\le1$, the payoff over a full four-cycle is
\[
  \sum_{k=0}^3v_{k+1}\cdot r=0.
\]
Thus external regret is zero.

Let $U$ be the unitary rotating the Bloch sphere by $\pi/2$ around the $z$-axis, and set $\Lambda(\rho)=U\rho U^\dagger$.  Then $\Lambda(\rho_k)=\rho_{k+1}$ and
\[
  \Tr(H_t\Lambda(\rho_t))=1
\]
for every $t$.  The coherent deviation earns one unit per round while the original sequence earns zero.
\end{proof}

\section{Computational remarks}

The regret and equilibrium guarantees are finite-time information-theoretic guarantees in an explicit convex-optimization model.  They do not say that each update is as cheap as a classical softmax, and they should not be read as a polylogarithmic-time quantum algorithm.  A round consists of solving a fixed-point feasibility problem and a relative-entropy convex program over the CPTP Choi body.

The fixed-point step is the feasibility problem
\[
  \rho\succeq0,
  \qquad
  \Tr\rho=1,
  \qquad
  \Lambda_t(\rho)=\rho.
\]
It is always feasible.  Approximate fixed points add the error described earlier.

The mirror step is a convex optimization problem over the spectrahedron $\mathcal C_d$ with a relative-entropy objective.  Its KKT conditions have the form
\[
  \log J_{t+1}=
  \log J_t+
  \eta A_t-I_{\rm out}\ot B_t-
  \gamma_t I,
\]
where $B_t$ enforces $\Tr_{\rm out}J_{t+1}=I/d$ and $\gamma_t$ enforces trace.  This is a noncommutative matrix-scaling problem.  In diagonal classical subalgebras it reduces to exponential weights for classical swap-regret reductions.  In the fully noncommutative case, efficient implementation depends on convex-optimization or matrix-scaling primitives; the theorem should be read in an oracle or finite-dimensional convex-optimization model.

The exploitability audit in Proposition~\ref{prop:sdp-audit}, by contrast, is a standard linear SDP because the candidate state $\omega$ is fixed and the deviation payoff is linear in the Choi state.  This is polynomial-time only in the explicit Choi dimension; for large $d_i$ it can still be computationally expensive.

\section{Scope and future directions}

The main equilibrium theorem is full-information: each player observes the local payoff effect induced by the other players' current states.  Appendix~\ref{app:probing-bandit} gives a first partial-feedback extension in a probing model, where the learner sometimes plays Haar-random pure states and receives one Bernoulli payoff sample.  The resulting guarantee is pseudo-regret, not the stronger pathwise quantity $\mathbb E\sup_\Lambda\sum_t\cdots$.  This distinction matters because coherent swap regret benchmarks transformations of the states actually played, while probing breaks the fixed-point identity on exploration rounds.  Payoff-based matrix multiplicative weights is known for external-regret quantum games under scalar feedback \cite{LotidisMertikopoulosBambosBlanchet2023}; a fully pathwise bandit coherent-swap theorem is still open.

The learned equilibrium is separable because uncoupled product play followed by time averaging produces a separable state.  The channel-proofness definition and SDP audit apply to entangled recommendation states as well, but learning such states requires a mediator, shared entanglement, communication, or another model that can actually generate entangled recommendations.

The present theory is for state-valued actions.  A natural extension is coherent internal regret for quantum decision kernels.  If a player's action is itself a channel $\Phi:X\to Y$, the appropriate deviations are superchannels mapping channels to channels.  The analogue of this paper would learn a superchannel, play a fixed channel of it, and compete with all superchannel rewrites of the realized policy.

Finally, the lower bound comes from a diagonal classical subgame.  This proves the optimal worst-case rate in the stated moderate-horizon regime, and the purity-sensitive upper bound shows how quantum structure can improve the upper bound on easier sequences.  What remains open is a genuinely noncommutative lower-bound theory for restricted families: low-rank or high-entropy fixed points, Pauli-diagonal games, bounded-Kraus deviations, non-entanglement-breaking comparator families, and superchannel deviations.

\section{Conclusion}

Coherent swap regret is the quantum internal-regret benchmark against local physical rewrites of a learner's realized states.  This paper gives the optimal worst-case rate in the oracle or finite-dimensional convex-optimization model, a purity-sensitive refinement, and the variance-collapse estimate that makes the Choi-body mirror analysis sharp.  It also shows that the resulting dynamics compute channel-proof separable quantum recommendations.  The audit examples show that replacement stability can miss profitable physical preprocessing even in small systems, and the online separations show that external regret can miss coherent deviations with linear gain.  For decentralized quantum strategic learning, the relevant stability test is therefore regret against local quantum operations.

\clearpage
\appendix
\leftappendixsections

\section{Payoff normalization}

For Hermitian payoffs $H_t$ with $\norm{H_t}_\infty\le L$, set
\[
  G_t=\frac{H_t+LI}{2L}.
\]
Then $0\preceq G_t\preceq I$ and
\[
  \Tr[G_t(\Lambda(\rho_t)-\rho_t)]
  =
  \frac{1}{2L}\Tr[H_t(\Lambda(\rho_t)-\rho_t)]
\]
for every CPTP map $\Lambda$.  Therefore effect-regret bounds imply norm-bounded payoff-regret bounds by multiplying by $2L$.

\section{Diagonal embedding of classical swap regret}

The lower bound uses the following exact embedding.  Diagonal density matrices are classical mixed actions.  Diagonal payoff effects are classical payoff vectors.  Classical modification rules $\phi:[d]\to[d]$ are the entanglement-breaking CPTP maps
\[
  \Lambda_\phi(\rho)=\sum_i\langle i|\rho|i\rangle\proj{\phi(i)}.
\]
Thus the coherent deviation payoff equals the classical swapped payoff.  Since an adversary may restrict to diagonal effects, every coherent-swap learner induces a classical swap-regret learner on this subproblem.

\section{Unbiased effect-estimate feedback}
\label{app:unbiased-feedback}

The main theorem is full-information because it assumes the learner observes the payoff effect $G_t$.  This appendix records the pseudo-regret statement that remains true when $G_t$ is replaced by an unbiased Hermitian estimate.  This is not a complete physical bandit theorem.  A one-shot measurement model must still supply estimators with the stated Choi-variance control, and a bound on the expected pathwise quantity $\mathbb E\sup_\Lambda \Reg_T(\Lambda)$ would require an additional uniform martingale argument.  The goal is to state the mathematical interface between coherent swap regret and partial-feedback implementations without overstating it.

\begin{lemma}[Signed second-order mirror inequality]
\label{lem:signed-second-order}
Let $\mathcal K\subseteq\D(\C^m)$ be compact and convex.  Let $P_1\in\mathcal K$ be positive definite and define
\[
  P_{t+1}=\arg\max_{P\in\mathcal K}\{\eta\ip{A_t}{P}-D(P\|P_t)\}.
\]
If $A_t=A_t^*$ and $\eta\norm{A_t}_\infty\le1$ for every $t$, then, for every $U\in\mathcal K$,
\[
  \sum_{t=1}^T\ip{A_t}{U-P_t}
  \le
  \frac{D(U\|P_1)}{\eta}
  +\eta\sum_{t=1}^T\ip{A_t^2}{P_t}.
\]
\end{lemma}

\begin{proof}
The proof is identical to Theorem~\ref{thm:second-order-md} except for the scalar inequality.  If $\eta\norm{A_t}_\infty\le1$, then the spectrum of $\eta A_t$ lies in $[-1,1]$.  The elementary bound $e^x\le1+x+x^2$ for $x\in[-1,1]$ gives, by functional calculus,
\[
  e^{\eta A_t}\preceq I+\eta A_t+\eta^2 A_t^2.
\]
The Golden-Thompson and Gibbs-variational steps then give the same second-order term.
\end{proof}

\begin{theorem}[Unbiased effect-estimate pseudo-regret]
\label{thm:unbiased-feedback}
Suppose that, after choosing $\rho_t$ and $J_t$, the learner receives a random Hermitian estimator $\widehat G_t$ satisfying
\[
  \mathbb E[\widehat G_t\mid\mathcal F_t]=G_t,
\]
where $0\preceq G_t\preceq I$ is the true payoff effect and $\mathcal F_t$ is the history before the feedback at round $t$.  Set
\[
  \widehat A_t=d(\widehat G_t\ot\rho_t^{\mathsf T})
\]
and update the Choi learner with $\widehat A_t$ instead of $A_t=d(G_t\ot\rho_t^{\mathsf T})$.  If $\eta\norm{\widehat A_t}_\infty\le1$ almost surely, then for every fixed comparator channel $\Lambda^\star$ with Choi state $J^\star$,
\[
  \mathbb E\sum_{t=1}^T
  \Tr\!\left[G_t\bigl(\Lambda^\star(\rho_t)-\rho_t\bigr)\right]
  \le
  \frac{2\log d}{\eta}
  +\eta\sum_{t=1}^T
  \mathbb E\,\ip{\widehat A_t^2}{J_t}.
\]
Equivalently,
\[
  \sup_{\Lambda^\star}\mathbb E\sum_{t=1}^T
  \Tr\!\left[G_t\bigl(\Lambda^\star(\rho_t)-\rho_t\bigr)\right]
  \le
  \frac{2\log d}{\eta}
  +\eta V_T,
\]
where
\[
  V_T:=\sum_{t=1}^T\mathbb E\,\ip{\widehat A_t^2}{J_t}.
\]
Thus any partial-feedback scheme that supplies unbiased effect estimates with Choi-variance budget $V_T$ obtains pseudo-regret at most $2\sqrt{2V_T\log d}$ when $\eta=\sqrt{2\log d/V_T}$ is feasible.  This statement does not by itself bound $\mathbb E\CReg_T$, where the supremum over comparator channels is inside the expectation.
\end{theorem}

\begin{proof}
Fix a comparator channel with Choi state $J^\star$.  Since $J_t$ and $\rho_t$ are $\mathcal F_t$-measurable,
\[
  \mathbb E\!\left[\ip{\widehat A_t}{J^\star-J_t}\mid\mathcal F_t\right]
  =
  \ip{A_t}{J^\star-J_t}.
\]
Expected regret against this fixed $J^\star$ is therefore the expected signed-gain regret of the Choi learner driven by $\widehat A_t$.  Applying Lemma~\ref{lem:signed-second-order} and using $D(J^\star\|I/d^2)\le2\log d$ gives the first display.  Taking the supremum outside the expectation gives the pseudo-regret form.
\end{proof}

\subsection{A probing-bandit instantiation}
\label{app:probing-bandit}

The previous theorem gives an interface.  We now give one concrete way to satisfy it, at the price of a worse dimension dependence and a pseudo-regret guarantee.  The learner occasionally sacrifices the fixed-point action and probes the unknown payoff effect with a Haar-random pure state.  This is a partial-feedback result, but it is not meant to be the final bandit theory.  It controls
\[
  \sup_{\Lambda}\mathbb E\sum_{t=1}^T
  \Tr\!\left[G_t(\Lambda(\sigma_t)-\sigma_t)\right],
\]
where $\sigma_t$ is the actually played state, not the stronger pathwise quantity $\mathbb E\sup_\Lambda\sum_t\cdots$.

\begin{theorem}[Probing-bandit coherent-swap pseudo-regret]
\label{thm:probing-bandit}
Assume a nonanticipating payoff sequence $0\preceq G_t\preceq I$: conditional on the history before round $t$, $G_t$ is fixed before the learner's current randomization.  At round $t$, let $\rho_t$ be the fixed point of the current learned channel.  With probability $1-\gamma$ the learner plays $\sigma_t=\rho_t$.  With probability $\gamma$, the learner draws a Haar-random pure state $\Pi_t=\ket{\psi_t}\!\bra{\psi_t}$, plays $\sigma_t=\Pi_t$, and observes one Bernoulli payoff sample $Y_t$ with
\[
  \mathbb E[Y_t\mid \Pi_t,G_t]=\Tr(G_t\Pi_t).
\]
On probe rounds define
\[
  \widetilde G_t=d(d+1)Y_t\Pi_t-dY_t I,
\]
and on non-probe rounds set $\widetilde G_t=0$.  The Choi learner is updated using
\[
  \widehat G_t=\frac{1}{\gamma}\widetilde G_t,
  \qquad
  \widehat A_t=d(\widehat G_t\ot\rho_t^{\mathsf T}).
\]
If $T\ge d^4\log d$ and
\[
  \gamma=\left(\frac{d^4\log d}{T}\right)^{1/3},
  \qquad
  \eta=\sqrt{\frac{\gamma\log d}{d^4T}},
\]
then the expected coherent-swap pseudo-regret against the actually played states satisfies
\[
  \sup_{\Lambda\in\CPTP(d)}
  \mathbb E\sum_{t=1}^T
  \Tr\!\left[G_t(\Lambda(\sigma_t)-\sigma_t)\right]
  \le
  6d^{4/3}T^{2/3}(\log d)^{1/3}.
\]
The same result holds if Haar-random pure states are replaced by an exact complex projective $2$-design.
\end{theorem}

\begin{proof}
We first verify unbiasedness.  For Haar-random $\Pi=\ket\psi\!\bra\psi$,
\[
  \mathbb E_\psi[\Tr(G\Pi)\Pi]
  =\frac{G+\Tr(G)I}{d(d+1)},
  \qquad
  \mathbb E_\psi[\Tr(G\Pi)]=\frac{\Tr G}{d}.
\]
Hence, on averaging over $\Pi$ and the Bernoulli sample $Y$,
\[
  \mathbb E[d(d+1)Y\Pi-dYI]=G.
\]
The additional Bernoulli exploration coin and the factor $1/\gamma$ therefore give
\[
  \mathbb E[\widehat G_t\mid\mathcal F_t]=G_t,
\]
where $\mathcal F_t$ is the history before the current randomization.  Thus Theorem~\ref{thm:unbiased-feedback} applies to the fixed-point sequence $\rho_t$ once we control the Choi variance.

For $M=d(d+1)Y\Pi-dYI$, $\norm M_\infty\le d^2$, so
\[
  \norm{\widehat A_t}_\infty\le \frac{d^3}{\gamma}.
\]
The displayed learning rate is feasible because $T\ge d^4\log d$ implies $\eta d^3/\gamma\le1$.  Moreover, since $Y^2=Y$ and $0\le\Tr(G\Pi)\le1$,
\begin{align*}
  \mathbb E[M^2]
  &\preceq
  \mathbb E_\psi[(d(d+1)\Pi-dI)^2]  \\
  &=d^2\mathbb E_\psi[(d^2-1)\Pi+I]
  \preceq 2d^3 I.
\end{align*}
Consequently, using $\Tr_{\rm out}J_t=I/d$,
\begin{align*}
  \mathbb E\bigl[\ip{\widehat A_t^2}{J_t}\mid\mathcal F_t\bigr]
  &\le
  \frac{d^2}{\gamma}
  \Tr\!\left[J_t\bigl(2d^3 I\ot\bigl(\rho_t^{\mathsf T}\bigr)^2\bigr)\right] \\
  &=\frac{2d^5}{\gamma}
  \Tr\!\left[(I/d)\bigl(\rho_t^{\mathsf T}\bigr)^2\right]
  \le \frac{2d^4}{\gamma}.
\end{align*}
The unbiased-feedback theorem gives pseudo-regret against the fixed-point sequence at most
\[
  \frac{2\log d}{\eta}+\eta\frac{2d^4T}{\gamma}.
\]

It remains to relate the fixed-point sequence to the actually played sequence.  For a fixed comparator $\Lambda$, the one-round deviation gain lies in $[-1,1]$.  If the learner exploits, it plays $\rho_t$; if it probes, it plays $\Pi_t$.  Therefore
\begin{align*}
  \mathbb E[\Tr G_t(\Lambda(\sigma_t)-\sigma_t)\mid\mathcal F_t]
  &\le
  \Tr G_t(\Lambda(\rho_t)-\rho_t)+2\gamma.
\end{align*}
Summing and taking the supremum over $\Lambda$ outside the expectation adds at most $2\gamma T$.

With the displayed parameters,
\[
  \frac{2\log d}{\eta}+\eta\frac{2d^4T}{\gamma}
  =4d^{4/3}T^{2/3}(\log d)^{1/3},
  \qquad
  2\gamma T=2d^{4/3}T^{2/3}(\log d)^{1/3}.
\]
Thus the total contribution is
\[
  4d^{4/3}T^{2/3}(\log d)^{1/3}
  +2d^{4/3}T^{2/3}(\log d)^{1/3}
  =6d^{4/3}T^{2/3}(\log d)^{1/3},
\]
which proves the stated bound.  The proof used only the first two Haar moments, so any exact complex projective $2$-design gives the same identities.
\end{proof}

\end{document}